\documentclass[conference]{IEEEtran}

\usepackage[dvips]{color}
\usepackage{epsf}
\usepackage{times}
\usepackage{epsfig}
\usepackage{graphicx}
\usepackage{amsmath}
\usepackage{amssymb}
\usepackage{amsxtra}
\usepackage{here}
\usepackage{rawfonts}
\usepackage{times}
\usepackage{url}
\usepackage{cite}
\usepackage{amsmath}
\usepackage[ruled,vlined]{algorithm2e}
\usepackage{epstopdf}
\usepackage{amsmath,amsthm}

\newtheorem{theorem}{\bf Theorem}

\newtheorem{lemma}{\bf Lemma}

\newlength{\aligntop}
\newlength{\alignbot}
\makeatletter
\renewenvironment{align}{%
  \vspace{\aligntop}
  \start@align\@ne\st@rredfalse\m@ne
}
{ %
  \math@cr \black@\totwidth@
  \egroup
  \ifingather@
    \restorealignstate@
    \egroup
    \nonumber
    \ifnum0=`{\fi\iffalse}\fi
  \else
    $$%
  \fi
  \ignorespacesafterend
 \vspace{\alignbot}\par\noindent
}

\begin{document}
\title{\huge Contract-Theoretic Resource Allocation \\ for Critical Infrastructure Protection\vspace{-0.4cm}}

\author{\IEEEauthorblockN{ AbdelRahman Eldosouky$^1$, Walid Saad$^1$, Charles Kamhoua$^2$, and Kevin Kwiat$^2$} \IEEEauthorblockA{\small
$^1$ Electrical and Computer Engineering Department, Virginia Tech, Blacksburg, VA, USA, Emails:\{iv727,walids\}@vt.edu \\
$^2$Air Force Research Laboratory, Information Directorate, Cyber Assurance Branch, Rome, NY, USA, \\Emails:\{charles.kamhoua.1,kevin.kwiat\}@us.af.mil\\
\vspace{-1cm}
}%
 \thanks{Approved for Public Release; Distribution Unlimited: 88ABW-2015-1367, dated 24 Mar 2015 \newline
This research was supported by the U.S. National Science Foundation under Grant NSF CNS-1446621 as well as by AFOSR via SFFP.}  
 } 
\date{}
\maketitle

\begin{abstract}
Critical infrastructure protection (CIP) is envisioned to be one of the most challenging security problems in the coming decade. One key challenge in CIP is the ability to allocate resources, either personnel or cyber, to critical infrastructures with different vulnerability and criticality levels. In this work, a contract-theoretic approach is proposed to solve the problem of resource allocation in critical infrastructure with asymmetric information. A control center (CC) is used to design contracts and offer them to infrastructures' owners. A contract can be seen as an agreement between the CC and infrastructures using which the CC allocates resources and gets rewards in return. Contracts are designed in a way to maximize the CC's benefit and motivate each infrastructure to accept a contract and obtain proper resources for its protection. Infrastructures are defined by both vulnerability levels and criticality levels which are unknown to the CC. Therefore, each infrastructure can claim that it is the most vulnerable or critical to gain more resources. A novel mechanism is developed to handle such an asymmetric information while providing the optimal contract that motivates each infrastructure to reveal its actual type. The necessary and sufficient conditions for such resource allocation contracts under asymmetric information are derived. Simulation results show that the proposed contract-theoretic approach maximizes the CC's utility while ensuring that no infrastructure has an incentive to ask for another contract, despite the lack of exact information at the CC.\end{abstract}

\section{Introduction}
\vspace{-0.1cm}Critical infrastructure (CI) is a term used to describe cyber and networking systems that are considered essential to the functioning of our modern economies and societies. CIs can cut across a variety of domains and applications. For instance, in the United States, the Department of Homeland Security classifies CIs into sixteen sectors that include energy production, financial services, communications, nuclear reactors, transportation systems, water supply, and financial services~\cite{CIP01}. However, in general there is no global classification for CIs and each country determines its own critical categories independently.

Critical infrastructure protection (CIP) has recently attracted significant attention~\cite{CIP03,CIP02,CIP04,CIP05}, particularly following recent terrorist and malicious attacks that targeted CIs across various countries. One main challenge in CIP is the fact that governmental agencies have only a limited amount of resources, such as personnel or even cyber resources, that can be used for CIP. Under such resource constraints and given the complex nature of CIs, it is imperative to develop practical resource management mechanisms that can optimally allocate such resources, given the criticality levels and vulnerabilities of the various CIs. Such resource deployment strategies are particularly critical for protecting CIs that are based in foreign countries or remote sites. In such scenarios, government agencies who want to protect local and foreign CIs will often own a control center (CC) that is responsible for monitoring these CIs and distributing resources among them. One major challenge for resource deployment here, is the fact that the CIs are often owned by different entities that \emph{consider their own CIs to be the most critical}. Indeed, every CI owner will report to the CC that its own infrastructure is the most vulnerable and most critical. Determining real levels of vulnerability and criticality of each individual infrastructure is very challenging for the CC. However, intuitively, the CC should design a proper mechanism to allocate resources based on the vulnerability and criticality levels of each CI. For example, highly vulnerable CIs should get higher resources than less vulnerable ones. Similarly, highly critical CIs must be properly prioritized in the CIP process. However, as each CI will attempt to get as much resources as possible by claiming that it is the most vulnerable or critical, the CC may not be able to properly distribute its limited resources.

The problem of CIP has attracted recent attention in the literature such as in ~\cite{CIP02,CIP03,CIP04,CIP05}. The work in~\cite{CIP03} focused on CI control systems by presenting a system to secure their functions and management tasks. Similarly, the authors in~\cite{CIP02} dealt with CI control systems by exposing and analyzing the vulnerabilities of these systems. In~\cite{CIP04}, the authors proposed a risk-aware robotic sensor network and applied it to CIP. The work in~\cite{CIP05} studied the vulnerabilities and protection challenges that face CIs and proposed a collaborative game theoretic solution for CIP. Although these works proposed solutions to CIP problem, they did not address resource allocation problem, as studied here.

The problem of resource allocation for CIP has been studied in recent works such as ~\cite{CIP09} and ~\cite{CIP10}. The work in~\cite{CIP09} studied an optimal resource allocation scenario in which resources were allocated to CIs depending on the likelihood of these CIs to be attacked according to their valuation on possible attackers. The authors in~\cite{CIP10} studied the problem of allocating resources where resources were supposed to protect an area and the objective was to maximize the area protected by these resources. Despite their importance, these existing works did not address the problem of asymmetric information in resource allocation which was first studied in~\cite{CIP11}. The problem discussed in~\cite{CIP11} cannot be generalized to a large-scale CIP system as it addresses an isolated problem and it does not take into account the impact of information availability on resource allocation.

In contrast to these works, here, we propose a contract-theoretic model to allocate resources for CIP under asymmetric information. Contract theory is a powerful framework from microeconomics that provides a useful set of tools for modeling mechanisms under information asymmetry ~\cite{CT00}. The key idea is that the CC should offer right contracts to CIs so that they have the incentive to truthfully reveal their information. In our model, the CC is seen as the principal that offers contracts to agents which are the CIs. While contract theory has been studied in the context of wireless networks \cite{CIP06,CIP07}, such works do not apply directly to CIP and their results cannot be generalized for accounting for criticality and vulnerability levels of the CIs.

The main contribution of this paper is to propose a resource allocation mechanism for CIP that can optimally allocate resources between a number of CIs without knowing their exact types. The problem is formulated using a contract-theoretic model in which the CC calculates the amount of resources that will be given to each CI and offer contracts to CIs with these values and the rewards they should reciprocate. In particular, we proposed a novel approach to define a CI by accounting for two different CI types according to the vulnerability and criticality levels. Both types are used in the process of resource allocation and the criticality level is also used to prioritize CIs that will be protected. For the formulated problem, we analyzed the necessary and sufficient conditions for deriving the optimal contracts. We studied the optimal contract and we proved that the problem has an optimal solution for the case of two CIs. The model was also shown to motivate each CI to reveal its actual type and accept the contract designed for its type, therefore allowing resource allocation in the absence of exact information at the CC on the criticality and vulnerability levels of the CIs. Simulation results show that the proposed approach will yield a higher CC utility when compared with a baseline resource allocation algorithm.

The rest of the paper is organized as follows.
Section~\ref{sec:sysmodel} provides the system model and defines the vulnerability and criticality levels of the CIs. In Section~\ref{sec:theory}, the problem is formulated as a contract-theoretic mechanism and several properties are derived and analyzed. Simulation results are discussed in Section~\ref{sec:sim}. Finally, conclusions are drawn in Section~\ref{sec:conclusion}.

\section{System Model}\label{sec:sysmodel}
We consider a system in which one control center (CC), that can represent a government agency is interested in sending missions to secure $N$ CIs in a set $\mathcal{N}$ that can be owned by different entities (e.g., foreign agencies, different department of defense agencies, etc.). The missions are viewed as \emph{resources} owned by the CC and that must be allotted to different CIs. Such resources can be personnel or cyber resources. Each CI has some vulnerable points that need to be protected. As the number of vulnerable points of a CI increases, the amount of resources needed to protect it will also increase. Infrastructures are classified into groups according to their vulnerability levels. The vulnerability level can be represented by an integer number $w_i$ where there are $M$ different levels in the set $\mathcal{M}$ and $M \leq N$; thus yielding different $M$ vulnerability levels. Infrastructures are grouped by an increasing order of vulnerability levels:
\vspace{-0.2cm}
\begin{equation}\label{eq:1}
w_1 < \ldots < w_i<\ldots < w_M.
\end{equation}

\vspace{-0.2cm}A higher $w$ implies that the CI has a higher vulnerability level. The CC \emph{does not have exact information} on the individual $w_i$ of every CI $i$. Instead, the CC can know with which probability a certain CI can belong to a certain $w$ type . Therefore, we let $p_{i,w_j}$ be the probability with which CI $i$ belongs to a certain type $w_j$.

Each CI has also a criticality level that can be represented by a number $\theta_i$ where there are different $K$ levels in the set $\mathcal{K}$ and $K \leq N$. Thus there exists $K$ criticality levels to which various CIs can belong. The criticality level is determined by factors such as the service performed by this CI and its relation to other CIs. The CIs are grouped by an increasing order of criticality levels:
\vspace{-0.2cm}
\begin{equation}\label{eq:2}
\theta_1 < \ldots < \theta_i<\ldots < \theta_K.
\end{equation}

\vspace{-0.2cm}A higher $\theta$ implies that the CI has a higher criticality level and thus it is more critical for the CC to protect it. Similar to the vulnerability levels here, we assume that the CC \emph{does not have exact information} on the individual $\theta_i$ of CI $i$. Instead, the CC can only know with which probability a certain CI can belong to $\theta$ type . Therefore, we let $q_{i,\theta_j}$ be the probability with which CI $i$ belongs to a certain type $\theta_j$. The criticality level is used mainly to help the CC decide which CIs will be protected in case not enough resources are available for all CIs.

The values of $w$ and $\theta$ are selected in a way that makes the resource allocation depends primarily on the vulnerability level. The criticality level affects the resource allocation but without superseding the vulnerability level, i.e., the criticality level will help the CI to get more resources than a less critical infrastructure but not more than a highly vulnerable one. Therefore, the values of $\theta$ when combined with $w$; should satisfy the following property:
\vspace{-0.3cm}\begin{equation}
\theta_K \cdot w_i \leq \theta_1 \cdot w_{i+1} , \forall i=1,\dots, {M-1}.
\end{equation}

\vspace{-0.2cm}In this case, $\theta$ can be seen as a sub-type under $w$ type in the process of allocation, although they are really independent.

To address the resource allocation problem, we explore the analogy between allocating resources to CIs and forming contractual agreements between firms and employees. We propose to use \emph{contract theory} -- a powerful framework from microeconomics~\cite{CT00,CIP08}, that allows to analyze the process of creating contracts between firms and employees. Here, we note that, although some recent works \cite{CIP06,CIP07} have looked at contract theory for wireless communication; however, these works do not handle the challenges of CI resource allocation.

We cast the CIP problem as a \emph{contractual situation with asymmetric hidden information} between a firm, here being the CC, and a number of employees, here being the $N$ CIs (or their owners). The asymmetric hidden information property stems from the fact that the CC does not know the exact vulnerability and criticality levels of every CI. To overcome this information asymmetry, the CC must properly specify a \emph{contract} defined as a pair $(T,R(T))$ where $T$ is the amount of resources allocated to the CI, which can be viewed as the reward/payment made by the firm to the employee and $R$ is the reward that the CC reaps when protecting this CI. We will assume in this model that the reward is an increasing linear function in resources $T$ that takes the form $R_i(T)=r_iT$ where $r_i$ is determined by the vulnerability type $w_i$ such that $r_i$ is higher with higher $w$'s. This implies that, for a higher vulnerability level, the CI is required to pay a higher reward than a less vulnerable one, if they take the same amount of resources. Actually, this reward function design is very important in order for the contract to be binding. By using this design, the CI that claims that has a higher vulnerability level to get more resources than it needs, will be required to pay a higher reward for the needed resources. The signing of a contract between the CC and a certain CI is thus an agreement by the CC to send certain resources to protect the CI which in return will pay a reward $R$ to the CC.

In this system, instead of offering the same contract to all of the CIs and wasting resources, the CC will attempt to offer different contract bundles that are designed in accordance with different types of $w$ and $\theta$ for the available CIs. For the CC, when it decides to protect a certain CI of type $i$, its utility function can simply be defined as the difference between reward and resources allocated multiplied by the CI type, i.e.,
\begin{equation}
U_\textrm{CC,i}(T_i) = \theta_i w_i (R_i(T_i) -  T_i).
\end{equation}
Since there are $M$ types of CIs according to type $w$ with probability $p_{i,w_j}$ and $K$ types according to $\theta$ with probability $q_{i,\theta_j}$, the total utility of the CC can be given by:
\begin{equation}
U_\textrm{CC}(T)=\sum\limits_{i \in \mathcal{N}}\Big(\sum\limits_{k \in \mathcal{K}}q_{i,\theta_k} \cdot \theta_k \Big)\Big (\sum\limits_{j \in \mathcal{M}}p_{i,w_j} \cdot w_j \cdot (R_j(T_i) - T_i)\Big).
\end{equation}

From the CI side, the utility function of a certain CI $i \in \mathcal{N}$:
\begin{equation}
U_{i}(T_i)=\theta_i w_i V(T_i) - \beta R_i(T_i),
\end{equation}
where $\beta$ is a positive unit cost parameter that is less than 1 and $V(T_i)$ is the evaluation function regarding the rewards (how much does this CI value the resources allocated) which is a strictly increasing function of $T$ that takes the form $V(T_i)=vT_i$ where $v$ is the numerical value for the evaluation function. Here, we assume that, to reward the CC, the CI has to pay some cost, such as a negotiation or implementation cost. The contract offered by the CC needs to be feasible for the CI, i.e., it needs to be persuading for the CI to accept. To this end, next, we discuss the conditions for the feasibility of a contract in the studied model.

\subsection{Feasibility of a Contract}
To ensure that both CC and CI owners have an incentive to work together for CIP, the contracts represented by the pairs $(T_i,R_i(T_i))$ must satisfy two key properties:
\begin{enumerate}
\item \emph{Individual Rationality (IR):} The contract that an infrastructure selects should guarantee that the utility of this infrastructure is nonnegative, i.e.,
\begin{equation}
U_i = \theta_i w_i V(T_i) - \beta R_i(T_i) \ge 0,\ i \in \mathcal{N}.
\end{equation}
\item \emph{Incentive Compatibility (IC):} Each infrastructure must always prefer the contract designed for its type, over all other contracts, i.e., $\forall i,\ j \in \mathcal{N},\ i\neq j$:
\begin{equation}
 \theta_i w_i V(T_i) - \beta R_i(T_i) \ge  \theta_i w_i V(T_j) - \beta R_j(T_j).
\end{equation}

\end{enumerate}

The IR constraint ensures that, when a CI signs a certain contract, the received reward must compensate the effort that the CI owner has exerted for the CC. The IC constraint allows to overcome the information asymmetry as it allows to satisfy the revelation principle~\cite{CT00}:  A certain CI of type $i$ will always prefer the contract $(T_i, R_i(T_i))$ that the CC designed for its type over all other possible contracts. In other words, a CI $i$ receives the maximum utility when selecting the contract designed for its own type and, thus, this CI will have an incentive to reveal its \emph{true vulnerability and criticality levels}. A contract is therefore said to be \emph{feasible} if both IR and IC are satisfied. We can state the following lemma from the previous conditions:
\begin{lemma}\label{lem1}
For any feasible contract $(T,R)$; $T_i > T_j$ if and only if $w_i > w_j$.
\end{lemma}

\begin{proof} We prove this lemma by using the IC constraint. we have:

\begin{align}
\theta_i w_i V(T_i)-\beta R_i(T_i) > \theta_i w_i V(T_j)-\beta R_j(T_j), \nonumber\\
\theta_j w_j V(T_j)-\beta R_j(T_j)>\theta_j w_j V(T_i)-\beta R_i(T_i).\nonumber
\end{align}

By adding the two inequalities, we get:
\begin{align}
\theta_i w_i V(T_i)+ \theta_j w_j V(T_j) > \theta_i w_i V(T_j)+\theta_j w_j V(T_i), \nonumber\\
\theta_iV(T_i)- \theta_j w_j V(T_i) > \theta_i w_i V(T_j)-\theta_j w_j V(T_j),\nonumber \\
V(T_i)(\theta_i w_i - \theta_j w_j) > V(T_j) (\theta_i w_i-\theta_j w_j). \nonumber 
\end{align}
\end{proof}
\vspace{-0.2cm}Since $w_i>w_j$ and this implies that $\theta_i w_i>\theta_j w_j$, we obtain $V(T_i) >V(T_j)$. By definition, we know that $V(T)$ is an increasing function of $T$, and therefore, since $V(T_i) >V(T_j)$ we have $T_i > T_j$.

Lemma~\ref{lem1} simply proves that the CC must provide more resources to the CI with higher number of vulnerability points, i.e., the one that belongs to a higher $w$ type. This essentially corroborates mathematically our intuition that more resources must be dedicated to more vulnerable CI. Using this lemma, we can state the following \emph{monotonicity} property:
\vspace{-0.1cm}
\begin{equation}
	T_i \leq T_j \ \textnormal{if} \  w_i < w_j  , \forall i,\ j \in \mathcal{N}.
\end{equation}

Another lemma that can be derived from the IR and IC constraints pertains to the utility of the CI:
\begin{lemma}\label{lem2}
For any feasible contract $(T,R(T))$, the utility of each infrastructure must satisfy:
\begin{equation}
	U_i(T_i) \geq U_j(T_j) \ \textnormal{if}\ w_i > w_j \ , \forall i, j \in \mathcal{N}.
\end{equation}
\end{lemma}
\begin{proof} This result can be shown as follows. We know that an infrastructure which asks for more resources should provide larger rewards to the CC, i.e., if $w_i > w_j$ then $T_i > T_j$ and also $R_i > R_j$. Then, if $w_i > w_j$, we have:
\begin{align*}
U_i(T_i)=\theta_i w_i V(T_i)-\beta R_i(T_i) > \theta_i w_i V(T_j) - \beta R_j(T_j)  \\ 
> \theta_j w_j V(T_j)-R_j(T_j) = U_j(T_j). 
\end{align*}
\end{proof}

\vspace{-0.2cm}Thus, a CI with a higher vulnerability level will receive more utility than one with a lower vulnerability level. From the IC constraint and the two shown lemmas, we can easily deduce the following. If a higher type CI selects a contract designed for a lower type, the less received resources will jeopardize this CI's utility.
Moreover, if a lower type CI selects a contract intended for a higher type, the gain in terms of resources acquired cannot compensate the cost that this CI must reciprocate to the CC. A CI can thus receive its maximum utility if and only if it selects the contract that can best fit its type.

Finally two more constraints must be imposed. First, the CC should take into account that the summation of all allocated resources should be equal to the maximum resources available at the control center:
	$\sum\limits_{i \in \mathcal{N}} T_i = T_{\textrm{max}}$.

Second, that every CI should get sufficient amount of resources to overcome its vulnerable points. That means every type $w_i$ should be associated with a minimum amount of resources. Therefore, each CI will have a minimum required resources according to its $w$ type. This can be expressed as:
	$T_i \geq  T_{i,min}$.

\section{Optimal Contracts}\label{sec:theory}
In this section, we first investigate how the CC can actually find its optimal contracts. In essence, given the hidden information, the only information available at the CC is $p_{i,w_j}$ and $q_{i,\theta_j}$. The goal of the CC is to design contracts that allow it to maximize the use of its resources and, thus, to maximize its utility by solving the following optimization problem:
\begin{equation}\label{eq:opt}
\max_T \sum\limits_{i \in \mathcal{N}}\Big(\sum\limits_{k \in \mathcal{K}}q_{i,\theta_k} \cdot \theta_k \Big)\Big (\sum\limits_{j \in \mathcal{M}}p_{i,w_j} \cdot w_j \cdot (R_j(T_i) - T_i)\Big),
\end{equation}
\vspace{-0.5cm}
\begin{align*}
&\textrm{ s.t. } \hspace{0.5cm} 
\theta_i w_i V(T_i) - \beta R_i(T_i) \ge 0,\ i \in N, \\  \nonumber
&\hspace{1.1cm} \theta_i w_i V(T_i) - \beta R_i(T_i) \ge  \theta_i w_i V(T_j) - \beta R_j(T_j) ,i\neq j, \\ \nonumber
&\hspace{1.1cm} T_i \geq  T_{i,min},\\ \nonumber
&\hspace{1.1cm} \sum\limits_{i \in \mathcal{N}} T_i = T_{\textrm{max}}. \nonumber 
\end{align*}

\vspace{0.15cm}The problem contains a large number of constraints. For instance, the IC constraints correspond to $N(N-1)$ equations. To overcome this, next, we develop a way to relax the problem and reduce the number of constraints to get a more simple problem that could be solved. The problem can be relaxed  using a technique inspired from the work in ~\cite{CIP08}.

\subsection{Relaxed Problem}
The incentive compatibility must to be relaxed because for every one of the CIs we need to define $N-1$ conditions. Therefore, we will now study the local IC constraints, which are, the downward local IC (DLIC) which corresponds to the relation between CIs $i$ and $i-1$. The other local IC is the upward local IC (ULIC) which corresponds to the relation between CIs $i$ and $i+1$. We can now prove the following:
\begin{theorem}\label{theo1}
With the IR satisfied, the local incentive constraints
\begin{equation}\label{eq:DLIC}
\theta_i w_i V(T_i) - \beta R_i(T_i) \ge  \theta_i w_i V(T_{i-1}) - \beta R_{i-1}(T_{i-1}), \\ 
\end{equation}
\begin{equation}\label{eq:ULIC}
\theta_i w_i V(T_i) - \beta R_i(T_i) \ge  \theta_i w_i V(T_{i+1}) - \beta R_{i+1}(T_{i+1}).
\end{equation}
for all $i \in \mathcal{N}$ are sufficient for global incentive compatibility.
\end{theorem}

\begin{proof} Note (\ref{eq:DLIC}) is called $DLIC(i)$ and (\ref{eq:ULIC}) is called $ULIC(i)$.
We begin by expressing $DLIC(i)$ and $DLIC(i-1)$ as follows:
\begin{align*}
&\theta_i w_i V(T_i) - \theta_i w_i V(T_{i-1}) \geq \beta (R_i(T_i) - R_{i-1}(T_{i-1})),  \\
& \theta_{i-1} w_{i-1} V(T_{i-1})-\theta_{i-1} w_{i-1} V(T_{i-2}) \geq \\
& \hspace{3.6cm} \beta (R_{i-1}(T_{i-1})-R_{i-2}(T_{i-2})).
\end{align*}
Then by adding $DLIC(i)$ and $DLIC(i-1)$ we get:
\begin{align*}
\theta_i w_i V(T_i)-\theta_i w_i V(T_{i-1})+\theta_{i-1} w_{i-1} V(T_{i-1})  \\ 
-\theta_{i-1} w_{i-1} V(T_{i-2}) \geq \beta (R_i(T_i)-R_{i-2}(T_{i-2})).
\end{align*}
and as $\theta_{i-1} w_{i-1} \leq \theta_i w_i$ we can replace it in the previous equation to yield:
\begin{equation}\label{eq:The.proof}
\theta_i w_i V(T_i)-\theta_i w_i V(T_{i-2}) \geq \beta (R_i(T_i)-R_{i-2}(T_{i-1})).
\end{equation}
However, (\ref{eq:The.proof}) is the IC constraint for CIs $i$ and $i-2$ which can be written as $IC(i,i-2)$. This means that $DLIC(i)$ and $DLIC(i-1)$ imply $IC(i,i-2)$. With the same principle we can show that $IC(i,i-1)$ and $DLIC(i-2)$
imply $IC(i,i-3)$, etc. Therefore, starting at $i = N$ and proceeding inductively downward until $i=2$, $DLIC(i)$ implies that $IC(i,j)$ holds for all $i \geq j$. A similar argument in the reverse direction establishes that $ULIC(i)$
implies $IC(i,j)$ for $i \leq j$.
\end{proof}

We can also reduce the IR constraints. There are a total of $N$ IR constraints must be satisfied.
Assume, without loss of generality, that the CI $1$ is from type $w_1$. By using the IC constraints and the IR constraint of the first CI, referred to by $IR(1)$, we have:
\begin{align}
\theta_i w_i V(T_i)-\beta R_i(T_i) &\geq \theta_i w_i V(T_1)-\beta R_1(T_1) \nonumber\\ 
&\geq \theta_1 w_1 V(T_1) -R_1(T_1) \geq 0. \label{eq:IR1}
\end{align}

Thus, if the first IR constraint of $w$ type-1 user is satisfied, all the other IR constraints will automatically hold. Therefore, we only need to keep the first IR constraints and reduce the others.

After reducing the constraints, we have a new problem which is the same as the problem in equation (\ref{eq:opt}) but with the new relaxed constraints in equations (\ref{eq:DLIC}) and (\ref{eq:IR1}) instead of the complete IR and IC constraints.
 
 Note that design parameters such as the reward function $R$, $\theta$, $w$, and $\beta$ should be adjusted by the CC to ensure that $IR(1)$ is satisfied.

\subsection{Solution of the Relaxed Problem}
To solve the relaxed problem, we first observe that there are now only $2N$ inequality constraints and one equality constraint. We can use Lagrangian analysis along with KKT conditions to solve the problem. The Lagrangian of the problem is:
\vspace{0.2cm}
\begin{align}
L(T,\lambda,\mu)&=\sum\limits_{i \in \mathcal{N}}\Big(\sum\limits_{k \in \mathcal{K}}q_{i,\theta_k} \theta_k \Big)\Big (\sum\limits_{j \in \mathcal{M}}p_{i,w_j} w_j (R_j(T_i) - T_i)\Big) \nonumber \\ 
&+\sum\limits_{i=2}^N \mu_i \Big(\theta_i w_i V(T_i)-\theta_i w_i V(T_{i-1})-\beta R_i(T_i) \nonumber \\
& +\beta R_{i-1}(T_{i-1}) \Big) +\mu_1(\theta_1 w_1 V(T_1)-\beta R_1(T_1)) \nonumber \\
& + \sum\limits_{i=1}^N \mu_{N+i} (T_i - T_{i,min}) +\lambda (T_{\textrm{max}} - \sum\limits_{i=1}^N T_i).
\end{align}

\vspace{0.3cm}We need to solve this Lagrangian with KKT conditions to find all $T$ values along with $\mu$ values and $\lambda$. The solution of this problem is not straightforward as the complexity increases with the number of CIs. Therefore, we will show the solution for only two CIs, to show that the problem has a feasible solution. For two CIs, the Lagrangian will be:
\begin{align*}
L(T,\lambda,\mu)&=p_{1,w1} w1 (r_1 T_1 - T_1) + p_{1,w2} w2 (r_2 T_1 - T_1) \\ 
&+ p_{2,w1} w1 (r_1 T_2 - T_2) + p_{2,w2} w2 (r_2 T_2 - T_2) \\
&+\mu_2(\theta_2 w_2 v T_2-\theta_2 w_2 vT_1-\beta r_2T_2 + \beta r_1T_1) \\
&+\mu_1(\theta_1 w_1 v T_1-\beta r_1T_1)+\mu_3 (T_1 - T_{1,min}) \\
&+ \mu_4 (T_2 - T_{2,min})+\lambda (T_{\textrm{max}} - T_1-T_2).
\end{align*}

The KKT conditions for this Lagrangian are the relaxed problem constraints along with:
\begin{align*}
& p_{1,w1} w1 (r_1 - 1) + p_{1,w2} w2 (r_2 - 1)+ \mu_1(\theta_1 w_1 v-\beta r_1) \\ 
& +\mu_2(\beta r_1-\theta_2 w_2 v)+\mu_3-\lambda=0.  \\
& p_{2,w1} w1 (r_1 - 1) + p_{2,w2} w2 (r_2 - 1)+\mu_2(\theta_2 w_2 v -\beta r_2) \\
& +\mu_4 - \lambda = 0.\\
& \mu_1(\theta_1 w_1 vT_1 -\beta r_1T_1)= 0.\\
& \mu_2 (\theta_2 w_2 (vT_2 - vT_1) -\beta( r_2 T_2 + r_1 T_1)=0. \\
& \mu_3 (T_1 - T_{1,min}=0. \\
& \mu_4 (T_2 - T_{2,min}=0. \\
& \mu_1,\mu_2,\mu_3,\mu_4 \geq 0. 
\end{align*}

This problem gives only one optimal solution which is $T_1 = T_{1,min}$ and $T_2=T_{\textrm{max}} - T_1$. Actually this solution is only feasible if the following condition is satisfied $T_{1,min}+T_{2,min} \leq T_{\textrm{max}}$. This implies that low vulnerability type CI will take its minimum required resources and the rest goes to the higher type CI. This result is not surprising as it is aligned with contract-theoretic results that study the contractual situation between a firm and two agencies (of two different types). ~\cite{CT00}. For the case of more than two CIs, the lower type CI will get its lower limit and the rest of resources will be allocated to higher types according to their probabilities in a way to maximize the CC's utility. 

\subsection{Practical Implementation}
Beside designing contracts, the CC needs to communicate with CIs, determine which CIs to protect, and sign contracts with them. We give the actual steps taken by the CC in this regard in Algorithm~\ref{algorithm:1}. The CC begins by having the initial information such as the set of vulnerability levels $\mathcal{M}$, the probability $p_{i,w_j}$ of that a CI $i$ will belong to each of the $M$ levels, the set of criticality levels $\mathcal{K}$ and the probability $q_{i,\theta_j}$ of that a CI $i$ will belong to each of the $K$ levels. The CC also knows the minimum amount of resources required to protect a CI in each of the vulnerability levels as well as the total amount of available resources. The CC, hence, declares that it will offer resources to protect some CIs and begins to receive requests from CIs that are willing to be protected and designs optimal contracts for them.
\begin{algorithm}[t]\label{algorithm:1}
\DontPrintSemicolon
  \caption{Optimized Contract implementation of CC for Resource Allocation}
  \KwIn{$\mathcal{M}, \mathcal{K}, p_{i,w_j},q_{i,\theta_j},T_{\textrm{max}}, T_{i,min} $ }
   \KwOut{$(T,R(T))$}
   1. CC declares its willingness to protect some infrastructures\;
   2. Receive request from infrastructures willing to be protected\;
   3.\textbf{Solve the optimal contract} for current infrastructures\;
   \uIf{The program has a solution, i.e, the avaialble resources are sufficient for all users}{ Contracts are ready, proceed to step 4\;
    }
    \Else
    {Remove the least critical infrastructure (begin with higher probability)\;
    return to step 3\;
    }
   4. \textbf{The CC Offers the contracts and waits for feedback}\;
   \uIf{All infrastructures accepted the offered contracts}{
   proceed to step 5\;
   }
  \Else
  {
	return to step 3, for any previously excluded infrastructures\;
  } 
   5. Sign contracts with infrastructures and allocate resources\;
\end{algorithm}

Algorithm~\ref{algorithm:1} shows the importance of the criticality level. When the CC is not able to protect all the CIs, it will discard some CIs depending on their criticality levels as the CC is more interested in protecting higher critical CIs. This is done by removing the least critical infrastructures from the process of designing contracts. However, as the CC only knows the probabilities of criticality levels, it will remove the one that belongs to the lower criticality level with a higher probability. The CC repeats this process until there is enough resources for the rest of CIs. 
When CIs receive contracts, they will evaluate them and inform the CC whether they are willing to accept a contract, i.e., receive resources and return reward. If not all CIs accept a contract, the CC will reconsider any CIs that were excluded due to lack of resources. After this process is finished, the CC will sign contracts with CIs and allocate resources.

\section{Simulation Results and Analysis}\label{sec:sim}

Simulations are used to evaluate the designed mechanism. For our simulations, we choose $3$ vulnerability levels and $4$ criticality levels. The number of available resources is set to $500$. The reward function increases by $3$ for each $w$ type. The evaluation function was assumed to be two times the resources. The lower bounds associated with $w$ types are set to 20, 60, and 100 respectively. First, we check the feasibility of our contract. We assume that all CIs ask for protection and all of them accept contracts offered by the CC. We calculate the CC's utility in case of using the proposed mechanism and in case of allocating resources equally between CIs. 

\begin{figure}[t]
  \centerline{\includegraphics[width=7cm]{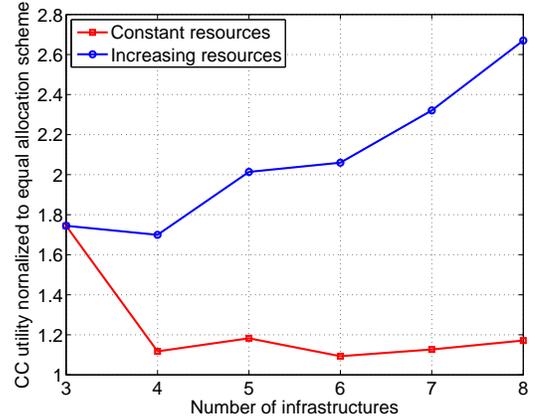}}
  \caption{The CC utility in the case of using the proposed contract and the case of equal resource allocation when fixing $T_{\textrm{max}}$ and when increasing $T_{\textrm{max}}$ by $30\%$ with every added infrastructure.}\label{fig:1}
\end{figure}

In Fig.~\ref{fig:1}, we show the variation in the CC utility as the number of infrastructures increases; the utility is normalized to the case of equal resource allocation. The figure studies two cases: fixing the amount of resources for all CIs and increasing the amount of resources each time a CI is added. When the number of resources is fixed to $500$, in Fig.~\ref{fig:1}, we can see that, with $3$ CIs there is about $75\%$ increase in the CC utility, relative to equal allocation. When more CIs are added, the percentage increase in CC utility is between $10\%$ and $20\%$. This is due to the fact that, when the number of CIs is small, the CC has more resources than needed and, thus, it will give them to higher types and to get higher rewards for the same resources. In the second case, the amount of resources increases by $30\%$ of the original amount each time we add a new CI. The CC utility in this case keeps increasing as the CC allocates the more available resources to higher types to get higher rewards. 

Next, we add a new vulnerability level with associated lower bound of 140 and we increase the number of available resources to 650. We have $4$ CIs, they are assumed to be within different $w$ types in ascending order, i.e. CI $1$ is within $w_1$ and so on. However, the CC still knows their probabilities not their actual types. Fig.~\ref{fig:2} shows the utility of CIs in case of using optimal contracts and in case of equal resource allocation. The figure proves the monotonicity property of the proposed contract as higher types CIs get higher utilities. We can also see in Fig.~\ref{fig:2} that, in case of optimal contracts, higher types CIs get higher utilities and lower types CIs get lower utilities compared to equal resource allocation. However, these lower utilities is not a problem as these CIs get much more resources than needed for protection.
\begin{figure}[t]
  \centerline{\includegraphics[width=7cm]{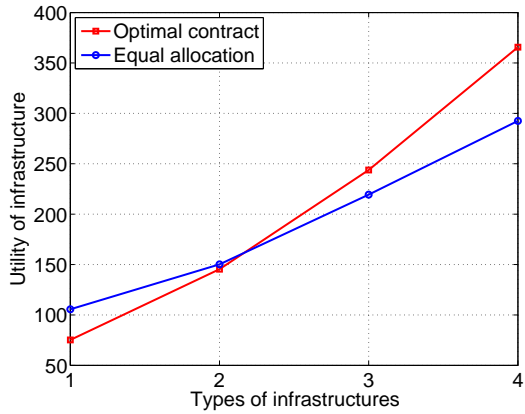}}
  \caption{Infrastructures' utilities in the case of using the proposed contract and the case of equal resource allocation.}\label{fig:2}
\end{figure}

In Fig.~\ref{fig:3}, while maintaining the parameters of Fig.~\ref{fig:2}, we show the utility of the infrastructure as the contract type varies. Here, we measure the utility of each CI if it used the contract designed for its type and contracts designed for other types. In Fig.~\ref{fig:3}, we can clearly see that it is better for every CI to use the contract designed for its type as this maximizes its utility. Actually, CIs can get more resources from choosing higher types contracts but will be required to pay higher rewards which is reflected in decreasing their utility.
\begin{figure}[t]
  \centerline{\includegraphics[width=7cm]{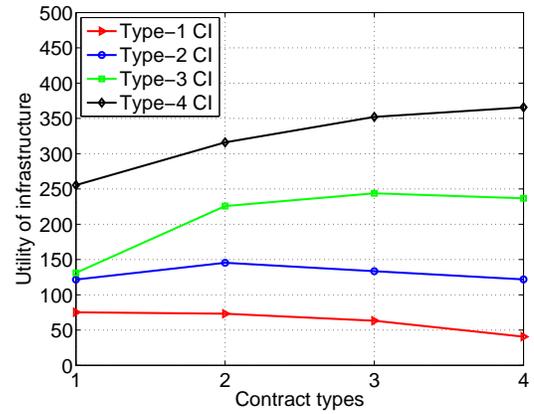}}
  \caption{The utility of each infrastructure while accepting the contract designed for his type or other contracts.}\label{fig:3}
\end{figure}

\section{Conclusions}\label{sec:conclusion}    

In this paper, we have studied the problem of resource allocation
for protecting CIs. We have formulated the problem using a contract-theoretic model in which a CC offers contracts
to a number of CIs and each one selects its best contract. For
each CI, we have defined two different types that correspond
to the vulnerability and the criticality levels. In the model, the
CC does not know these exact levels but only knows with
which probability a CI will belong to a certain level. We have
provided the necessary and sufficient conditions for such resource
allocation contracts under asymmetric information. The problem
was then relaxed and solved to show that it has an optimal
solution and motivates each CI to accept the contract designed
for each type. Simulation results have shown that this model
helps the CC to get higher utility than the case of equal resource
allocation. In addition, our results show that each CI will not
gain from selecting other contracts as its utility will not increase.         

\def\baselinestretch{0.8}
\bibliographystyle{IEEEtran}
\bibliography{references}

\end{document}